\documentclass[%
 aip,
 amsmath,amssymb,
 reprint,%
]{revtex4-1}

\usepackage{graphicx}
\usepackage{dcolumn}
\usepackage{bm}
\usepackage{amsthm}
\usepackage{amsmath,amscd}
\usepackage{mathrsfs}
\usepackage{here}
\usepackage{color}
\usepackage[hang,small,bf]{caption}
\usepackage[subrefformat=parens]{subcaption}
\newtheorem{theo}{Theorem}[section]
\newtheorem{defi}{Definition}[section]

\begin{document}

\preprint{AIP/123-QED}

\title{Infinite Ergodicity that Preserves the Lebesgue Measure}

\author{Ken-ichi Okubo}
\email{okubo@ist.osaka-u.ac.jp}
\altaffiliation[Present adress:]{Department of Information and Physical Sciences, 
Graduate School of Information Science and Technology, Osaka University 1-5 Yamadaoka, Suita, Osaka 565-0871 Japan}
\author{Ken Umeno}%
\email{umeno.ken.8z@kyoto-u.ac.jp}
\affiliation{%
 Department of Applied Mathematics and Physics, 
	Graduate School of Informatics, Kyoto University, Yoshida-honmachi, Sakyo-ku, Kyoto 606-8501 Japan
}%

\date{\today}

\begin{abstract}
We proved that for the countably infinite number of one-parameterized one dimensional dynamical systems, they preserve the Lebesgue measure and they are ergodic for the measure (infinite ergodicity).  Considered systems connect the parameter region in which dynamical systems are exact and the parameter region in which systems are dissipative, and correspond to the critical points of the parameter in which weak chaos occurs (the Lyapunov exponent converges to zero). These results are the generalization of the work by R. Adler and B. Weiss. We show that the distributions of normalized Lyapunov exponent for these systems obey the Mittag-Leffler distribution of order $1/2$ by numerical simulation.
\end{abstract}

\maketitle

\begin{quotation}
As a typical characteristics of ergodicity, the equality of the time average and the space average is pointed out.
However, there exist phenomena in which the time average is not equivalent to the space average in infinite ergodic systems \cite{akimoto2010subexponential}. The Boole transformation is known as a one dimensional map \cite{adler1973ergodic}
, which is proven that the transformation preserves the Lebesgue measure (infinite measure) and is ergodic.
In this paper it is proven that countably infinite number of one parameterized one dimensional maps which are generalized from the Boole transformation exactly preserve the Lebesgue measure and are ergodic at certain parameters. 
Additionally we show that in these maps the normalized Lyapunov exponent obeys the Mittag-Leffler distribution of order $1/2$
as well as the Boole transformation.
\end{quotation}

\section{Introduction}
Chaos theory has developed statistical physics through ergodic theory.
In a chaotic dynamics, it is difficult to predict future orbital state from past information
because the system is unstable, which is characterized by the sensitivity to initial conditions. However, from its mixing property, one can characterize the system statistically using the invariant density function.
The relation between microscopic dynamics and density function is important when 
macroscopic properties are led from microscopic dynamics, and ergodicity plays a
significant role in this derivation. 

In the case of a dynamical system $(X, T, \mu)$ with a normalized ergodic invariant measure $\mu$
where $X$ and $T$ represent the phase space and a map, respectively, for an observable 
$f \in L^1(\mu)$, a time average $\lim_{n\to\infty}\frac{1}{n}\sum_{i=0}^{n-1}f\circ T^i(x)$
converges to the phase average $\int_X f d\mu$ in almost all region \cite{birkhoff1931proof}.
In systems with a normalized ergodic measure, we can characterize their stability by Lyapunov exponent $\lambda$, which is
defined as $\lambda \overset{\rm def}{=} \lim_{n\to \infty}\frac{1}{n}\sum_{i=0}^{n-1}\log\left|T'(x_i)\right|$ 
when $\log\left|T'(x_i)\right|$ is a $L^1$ class function
for the measure $\mu$.
Usually, since it is difficult to obtain the information at infinite time, we use numerical simulations or apply the 
ergodicity to calculate the Lyapunov exponent as $\lambda = \int_X \log\left|T'(x)\right|d\mu$.
For example, for logistic map $x_{n+1}= ax_n(1-x_n)$, the Lyapunov exponent $\lambda$ at $a=4$ is obtained as
$\lambda = \int_0^1 4(1-2x)\frac{dx}{\pi\sqrt{x(1-x)}}=\log2$ \cite{jakobson1981absolutely}, and for the generalized Boole transformations 
$x_{n+1} = \alpha x_n -\frac{\beta}{x_n}, \alpha>0, \beta>0$ \cite{aaronson1997introduction,umeno2016exact}, $\lambda$ is obtained analytically as
$\lambda = \int_{-\infty}^\infty \left(\alpha + \frac{\beta}{x^2}\right)\frac{\sqrt{\beta(1-\alpha)}}{\pi\left\lbrace x^2(1-\alpha)+\beta\right\rbrace}dx = \log \left(1+2\sqrt{\alpha(1-\alpha)}\right)$
for $0<\alpha<1$. We know other systems whose invariant ergodic measure can be expressed explicitly \cite{umeno1998superposition,umeno2016ergodic}.

On the other hand, how about the case of infinite ergodic case? Consider the Boole transformation $x_{n+1}=T(x_n)=x_n-1/x_n$ \cite{boole1857xxxvi,adler1973ergodic},
which corresponds to $\alpha=\beta=1$ for the generalized Boole transformations where the dynamical system preserves the Lebesgue measure as an infinite ergodic measure. 
That means it holds that $\int_{-\infty}^\infty f(x) dx = \int_{-\infty}^\infty f\left(x-\frac{1}{x}\right)dx$ where
$f$ is a $L^1$ function with respect to $dx$.
For an observable $\log\left|T'\right|$, although the usual time average $\lim_{n\to\infty}\frac{1}{n}\sum_{i=0}^{n-1}\log\left|T'(x_i)\right|$ converges to zero \cite{akimoto2010subexponential}, the phase average is as 
$\int_{-\infty}^\infty \log\left|1+1/x^2\right|dx=2\pi$ \cite{akimoto2010subexponential},
so that the time average does not consistent with the phase average.

In infinite ergodic systems, the Darling, Kac and Aaronson theorem says that
if the observable $f$ is positive $L^1$ function in terms of invariant measure $\mu$,
the time average using the return sequence $a_n$ converges in distribution \cite{aaronson1997introduction}.
In the case of the Boole transformation, by defining the return sequence $a_n \overset{\rm def}{=} \frac{\sqrt{2n}}{\pi}$,
the distribution of $\frac{1}{a_n}\sum_{i=0}^{n-1}\log\left|T'(x_i)\right|$
converges to the Mittag-Leffler distribution of order $1/2$ \cite{aaronson1997introduction}.
That is, interesting phenomena are observed which are different from the usual ergodic theory and the standard statistical mechanics.

In infinite measure system, it is known that following $L^1$ class observables
converge to the Mittag-Leffler distribution such as Lempel-Ziv complexity \cite{shinkai2006lempel}, 
the transformed observation function for the correlation function \cite{akimoto2007new}, 
normalized Lyapunov exponent \cite{akimoto2010subexponential}, 
normalized diffusion coefficient \cite{akimoto2010role}
and that non-$L^1$ class observables such as time average of position \cite{korabel2012infinite}
converges to generalized arc-sin distribution 
\cite{thaler2002limit,thaler2006distributional,akimoto2008generalized}
or other distribution \cite{akimoto2015distributional}.

Infinite densities were observed in the context with the long time limit of solution
of Fokker-Planck equation for Brownian motion \cite{kessler2010infinite,dechant2011solution} and semiclassical Monte Carlo simulations of cold atoms \cite{holz2015infinite}.

In order to characterize the instability of systems with infinite measure, several quantities were invented such as
Lyapunov pair \cite{akimoto2010subexponential} and generalized Lyapunov exponent  \cite{korabel2009pesin,korabel2010separation,korabel2013numerical}.

In relation to Lyapunov exponent, the change of stability of systems characterizes their dynamical properties and is important phenomenon.
In particular, critical phenomena at which systems become unstable from stable called as
routes to chaos has attracted a lot of interests in the fields of Hamiltonian dynamics \cite{hioe1987stability}, intermittent systems \cite{manneville1979intermittency,pomeau1980intermittent,ott1994blowout,lamba1994scaling}, logistic map \cite{huberman1980scaling}, a differential equation \cite{milosavljevic2017analytic}, 
coupled  chaotic  oscillators \cite{liu2003universal}, a noise-induced system \cite{crutchfield1981scaling}, experiments(Belousov-Zhabotinskii reaction, Rayleigh-B\'enard convection, and Couette-Taylor flow) \cite{swinney1983observations} and optomechanics \cite{he2004critical,coillet2014routes,bakemeier2015route}.

For generalized Boole transformations, at the onset of chaos the Lyapunov exponent
defined by the time average, converges to zero as $\alpha \to 1$. The point $\alpha_c = 1$ is referred to as the critical point at which \textit{Type} 1 intermittency occurs \cite{umeno2016exact}. Since the parameter dependence at the critical point \textit{diverges} as $\lim_{\alpha \to 1-0}\left|\frac{\partial}{\partial \alpha}\left(1+2\sqrt{\alpha(1-\alpha)}\right)\right|=\infty$, we know that it is difficult to obtain the correct Lyapunov exponent by numerical experiments. The bahavior at $\alpha=\beta=1$ is explained by the Boole transformation in which the Lyapunov exponent derived from the time average does not consist with that derived from the phase average although the system is ergodic.

The authors previously extended the generalized Boole transformations by defining countably infinite number of one-parameterized maps, which are called super generalized Boole (SGB) transformations  \cite{okubo2018universality} and showed that the Lyapunov exponent converges to zero from positive value as $\alpha\to 1$ and \textit{Type} 1 intermittency occurs at $\alpha=1$ for countably infinite number of maps (SGB). 
The authors proved that the SGB transformations are \textit{exact} when $(K, \alpha)$ are in Range A. However, the ergodic property at $\alpha=\pm1$ of SGB ($K \geq 3$) was left unresolved.
In this paper, we show that all the super generalized Boole (SGB) transformations at $\alpha=\pm 1$ also preserve the Lebesgue measure and is proven to be ergodic as same as the Boole transformation. 
If we look at the foundation of statistical mechanics, the Liouville measure on $\mathbb{R}^{2N}$ is vitally important and this can  be regarded as the Lebesgue masure which is invariant under Hamiltonian dynamical system with $N$ degrees of freedom \cite{landau1970statistical,gallavotti2014nonequilibrium}.
Thus, it is of great interest to investigate  ``\textit{ergodic}'' Lebesgue measure on $\mathbb{R}$ which is invariant under nonlinear transformations from the physical point of view.

\section{Super Generalized Boole transformations}
In this section, let us define the super generalized Boole (SGB) transformations \cite{okubo2018universality}. 
At first, define a function $F_K:\mathbb{R}\backslash A \to \mathbb{R}\backslash A$
such as 
\begin{equation}
F_K(\cot\theta)\overset{\rm def}{=} \cot (K\theta)
\end{equation}
where $K\in \mathbb{N}\backslash \{1\}$ and $A$ represents a set of point $x \in\mathbb{R}$ such that
for finite iteration $n\in \mathbb{Z}$, $F_K^n(x)$ reaches the singular point.

$F_K$ corresponds to the $K$-angle formula of cot function.
For example, $F_2(x) = \frac{1}{2}(x-\frac{1}{x})$ corresponds to the 
$\cot(2\theta)=\frac{1}{2}\left(\cot\theta-\frac{1}{\cot\theta}\right)$ \cite{umeno2016ergodic}.

Then, super generalized Boole transformations $S_{K, \alpha}: \mathbb{R}\backslash B \to \mathbb{R}\backslash B$ are defined as follows.
\begin{equation}
x_{n+1}=S_{K, \alpha}(x_n) \overset{\rm def}{=} \alpha K F_K(x_n),
\end{equation}
where $|\alpha|>0$, $K\in \mathbb{N}\backslash\{1\}$ and 
$B$ represents a set of point $x \in\mathbb{R}$ such that
for finite iteration $n\in \mathbb{Z}$, $S_{K,\alpha}^n(x)$ reaches the singular point.

Let us define Range A as
$``0<\alpha<1 ~\mbox{and}~ K=2N" ~\mbox{or}~ ``\frac{1}{K^2}<\alpha<1 ~\mbox{and}~ K=2N+1"$
where $N\in \mathbb{N}$. 
When $(K, \alpha)$ are in Range A,
the super generalized Boole (SGB) transformations are \textit{exact} and when $\alpha>1$, the any orbits
diverge to the infinity and the SGB transformations do not preserve measure over real line \cite{okubo2018universality}. 

In the following one can extend the Range A to the newly defined Range B such that $``0<|\alpha|<1 ~\mbox{and}~ K=2N" ~\mbox{or}~ ``\frac{1}{K^2}<|\alpha|<1 ~\mbox{and}~ K=2N+1"$ where $N\in \mathbb{N}$. 
Let us define the Range A' as
\begin{equation*}
\left\lbrace
    \begin{array}{ccc}
         -1<\alpha<0 &\mbox{in the case of}& K=2N,  \\
         -1<\alpha<-\frac{1}{K^2}& \mbox{in the case of}& K=2N+1,
    \end{array}
\right.
\end{equation*}
where $N\in \mathbb{N}$. 
In the following extension from $\alpha$ to $|\alpha|$ can be proven in the similar way as the reference 37.
If the density function at the time $n$ ($f_n(x)$) is denoted as 
\begin{equation*}
    f_n(x) = \frac{1}{\pi}\frac{\gamma}{x^2+\gamma^2},
\end{equation*}
$f_{n+1}(x)$ is given by 
\begin{equation*}
    f_{n+1}(x) = \frac{1}{\pi}\frac{|\alpha|KG_K(\gamma)}{x^2+|\alpha|^2K^2G_K^2(\gamma)}
\end{equation*}
according to the Perron-Frobenius equation where $G_K(\gamma)$ corresponds to the the K-angle formula of the $\coth$ function\cite{okubo2018universality}. Then the scale parameter $\gamma$ is changed in a single iteration as 
\begin{equation*}
    \gamma \mapsto |\alpha|K G_K(\gamma).
\end{equation*}
Then, by changing the parameter from $\alpha$ to $|\alpha|$, we can prove straightforwardly that the SGB transformations $\{S_{K, \alpha}\}$ preserve the Cauchy distribution and the scale parameter can be chosen uniquely
when the parameters $(K, \alpha)$ are in Range B.
In terms of exactness, it holds that $\bar{S}_{K,\alpha}'(\theta)<0$ when $(K, \alpha)$ are in the Range A'. Then, $\bar{S}_{K,\alpha}(\theta)$ is also the monotonic function. Thus, we can prove that the SGB transformations $\{S_{K, \alpha}\}$ are exact when the parameters $(K, \alpha)$ are in Range B by considering the intervals $\{I_{j,n}\}$.
In the case of $\alpha<-1$, one straightforwardly sees that orbits diverge to the infinity and the SGB transformations do not preserve measure over real line.

From above discussion, we know that the SGB transformations are \textit{exact}
when the parameters $(K, \alpha)$ are in range B and the systems are disspative when $|\alpha|>1$. However,
the ergodic property at the critical point $\alpha=\pm 1$ was unsettled.

Then, what happens at $\alpha=\pm 1$? Since the statistical properties drastically change 
before and after the value of $\alpha=\pm 1$, the erogidic property of the \textit{critical} SGB transformations at
$\alpha=\pm 1$ is important. As we know, the Boole transformation which corresponds to the case of
$K=2, \alpha=1$ preserves the Lebesgue measure and are ergodic \cite{adler1973ergodic}.
In the following section, we show that all the SGB transformations at $\alpha= \pm 1$ preserve the Lebesgue
measure  for any $K \in \mathbb{N}\backslash\{1\}$.
Table \ref{Tab:form of SGB} shows the explicit 
form of $S_{K, \pm 1}$ for $K=2,3,4,5$ and $6$.

\begin{figure}[htb]
    \centering
    \includegraphics[width=\columnwidth]{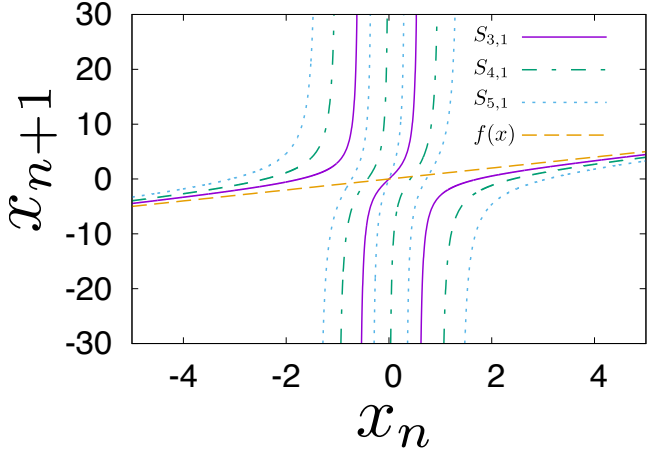}
    \caption{Return maps of $S_{3,1}, S_{4,1}$ and $S_{5,1}$. The function $f(x)=x$ represents
    the set of fixed points.}
    \label{Fig:Return map}
\end{figure}

\begin{table*}
\caption{\label{Tab:form of SGB} $S_{K, \pm1}(x)$ for $K=2,3,4,5$ and $6$.}
\begin{ruledtabular}
\begin{tabular}{ccccccc}
& $K=2$& 3 & 4 & 5 & 6\\
\hline
$S_{K, \pm1}(x)$& ${\displaystyle \pm\left(x-\frac{1}{x}\right)}$ &
${\displaystyle \pm 3 \frac{x^3-3x}{3x^2-1}}$ &
${\displaystyle \pm 4 \frac{x^4-6x^2+1}{4x^3-4x}}$ &
${\displaystyle \pm 5 \frac{x^5-10x^3+5x}{5x^4-10x^2+1}}$ &
${\displaystyle \pm 6 \frac{x^6-15x^4+15x^2-1}{6x^5-20x^3+6x}}$\\
\end{tabular}
\end{ruledtabular}
\end{table*}

\section{Infinite ergodicity for $\alpha=1,-1$}

\begin{theo}
	The SGB transformations at $\alpha=\pm 1$ preserve the Lebesgue measure.
\end{theo}

\begin{proof}
	
	The goal is to prove that
	\begin{equation}
	\left|S_{K, \pm 1}^{-1}I\right| = \left|I\right| \label{Ap Eq lebesgue}
	\end{equation}
	for any interval $I$ where $\left|\cdot\right|$ denotes the length of a interval.
	It is sufficient to verify this for intervals of $I=(0, \eta), \eta>0$ and $I=(\eta, 0), \eta<0$ \cite{adler1973ergodic}.
	
	(I) In the case of $\alpha=1$.
	
	(In the following, we will prove Eq. \eqref{Ap Eq lebesgue} holds for $\eta>0$; the proof for $\eta<0$ is similar.)
	The map $S_{K, 1}$ increases monotonically. We have that
	\begin{equation}
	x_{n+1} = S_{K, 1}(x_n) 
	\end{equation}
	and for $x_{n+1}=0$, $x_n = \cot\theta_n, \theta_n \in {\rm arccot}(\mathbb{R}\backslash B) \subset [0, \pi]$ satisfies the following relation:
	\begin{equation}
	\begin{array}{rll}
	K \theta_n &=& \displaystyle \frac{\pi}{2}+m\pi, m \in \mathbb{Z}\\
	\theta_n &=& \displaystyle\frac{\pi}{2K}+ \frac{m}{K}\pi,\\
	0&\leq& \frac{\pi}{2K}+ \frac{m}{K}\pi \leq \pi.
	\end{array}
	\end{equation}
	The range of possible values for $m$ is $m=0, 1,2, \cdots, K-1$.
	Then for $x_n$ such that $x_{n+1}=0$, it follows that
	\begin{equation}
	x_n = \cot\left(\frac{\pi}{2K} + \frac{m}{K}\pi\right), m=0,1,2,\cdots,K-1.
	\end{equation}
	For $x_n$ such that $x_{n+1}=K\cot(K\theta_n)=\eta$, it follows that
	\begin{equation}
	\begin{array}{lll}
	K\theta_n &=& \cot^{-1}\left(\frac{\eta}{K}\right) + m\pi\\
	\theta_n &=& \frac{1}{K}\cot^{-1}\left(\frac{\eta}{K}\right) + \frac{m}{K}\pi,\\
	0&\leq& \frac{1}{K}\cot^{-1}\left(\frac{\eta}{K}\right) + \frac{m}{K}\pi \leq \pi.
	\end{array}
	\end{equation}
	Here, since 
	\begin{equation}
	0<\cot^{-1}\left(\frac{\eta}{K}\right)<\frac{\pi}{2},
	\end{equation}
	the range of possible values for $m$ is given by
	\begin{equation}
	-\frac{1}{2}<-\frac{1}{\pi}\cot^{-1}\left(\frac{\eta}{K}\right) \leq m \leq K - \frac{1}{\pi}\cot^{-1}\left(\frac{\eta}{K}\right)<K,
	\end{equation}
	that is $m=0,1,2,\cdots, K-1$. Then $\theta_n$ and $x_n$ are given by
	\begin{equation}
	\begin{array}{lll}
	    \theta_n &=& \frac{1}{K}\cot^{-1}\left(\frac{\eta}{K}\right) + \frac{m}{K}\pi, m= 0,1,2,\cdots,K-1,\\
	    x_n &=& \cot\left\lbrace\frac{1}{K}\cot^{-1}\left(\frac{\eta}{K}\right) + \frac{m}{K}\pi\right\rbrace
	\end{array}
	\end{equation}
	where $\eta = K\cot(K\theta_n)$.
	Because the $S_{K,1}$ increases monotonically and the $\cot$ function decreases monotonically for $\theta\in[0, \pi]$,
	the interval that is mapped from $(0, \eta)$ by $S_{K, 1}^{-1}$ is 
	\begin{equation}
	\bigcup_{m=0}^{K-1}\left(\cot\left(\frac{\pi}{2K}+\frac{m}{K}\pi\right), \cot\left\lbrace \frac{1}{K}\cot^{-1}\left(\frac{\eta}{K}\right)
	+\frac{m}{K}\pi\right\rbrace \right).
	\end{equation}
	Then we have that
	\begin{equation}
	\begin{array}{lll}
	&&\left|S_{K, 1}^{-1}(0, \eta)\right|\\
	 &=& \displaystyle\sum_{m=0}^{K-1}\left[
	\cot\left\lbrace \frac{1}{K}\cot^{-1}\left(\frac{\eta}{K}\right)+\frac{m}{K}\pi\right\rbrace -
	\cot\left(\frac{\pi}{2K}+\frac{m}{K}\pi\right)
	\right].
	\end{array}
	\end{equation}
	In the following discussion, we consider $\displaystyle\sum_{m=0}^{K-1}\cot\left(\frac{\pi}{2K}+\frac{m}{K}\pi\right)$.
	
	(i) Case $K=2N$.
	For $\displaystyle\sum_{m=0}^{K-1}\cot\left(\frac{\pi}{2K}+\frac{m}{K}\pi\right)$, adding the terms corresponding to
	$m=0$ and $m=K-1$, we obtain
	\begin{equation}
	\begin{array}{lll}
	&&\cot\left(\frac{\pi}{2K}\right) + \cot\left(\frac{\pi}{2K}+\frac{K-1}{K}\pi\right)\\
	&=&
	\cot\left(\frac{\pi}{2K}\right) + \cot\left(\pi-\frac{\pi}{2K}\right) =0.
	\end{array}
	\end{equation}
	Adding the terms corresponding to $m=l$ and $m=K-1-l$, $l=0,\cdots,\frac{K}{2}-1$, we obtain
	\begin{equation}
	\cot\left(\frac{(2l+1)\pi}{2K}\right)+\cot\left(\pi-\frac{(2l+1)\pi}{2K}\right) =0.
	\end{equation}
	Thus, for $K=2N$, the following relation holds:
	\begin{equation}
	\sum_{m=0}^{K-1}\cot\left(\frac{\pi}{2K}+\frac{m}{K}\pi\right)=0.
	\end{equation}
	
	\begin{widetext}
	(ii)Case $K=2N+1$.
	We have
	\begin{equation}
	\begin{array}{lll}
	\displaystyle\sum_{m=0}^{K-1}\cot\left(\frac{\pi}{2K}+\frac{m}{K}\pi\right) &=&
	\displaystyle\sum_{m=0}^{\frac{K-3}{2}}\cot\left(\frac{\pi}{2K}+\frac{m}{K}\pi\right)
	+\cot\left(\frac{K-1+1}{2K}\pi\right)+
	\sum_{m=\frac{K+1}{2}}^{K-1}\cot\left(\frac{\pi}{2K}+\frac{m}{K}\pi\right)\\
	&=&\displaystyle\sum_{m=0}^{\frac{K-3}{2}}\cot\left(\frac{\pi}{2K}+\frac{m}{K}\pi\right)+
	\sum_{m=\frac{K+1}{2}}^{K-1}\cot\left(\frac{\pi}{2K}+\frac{m}{K}\pi\right).
	\end{array}
	\end{equation}
	\end{widetext}
	Much as in (i), because the term corresponding to $m=l$ negates the term corresponding to
	$m=K-1-l$, $l=0,\cdots,\frac{K-3}{2}$, it follows that
	\begin{equation}
	\displaystyle\sum_{m=0}^{K-1}\cot\left(\frac{\pi}{2K}+\frac{m}{K}\pi\right) = 0.
	\end{equation}
	Thus, we have that
	\begin{equation}
	\left|S_{K, 1}^{-1}(0, \eta)\right| = \sum_{m=0}^{K-1}\cot\left\lbrace \frac{1}{K}\cot^{-1}\left(\frac{\eta}{K}\right)+\frac{m}{K}\pi\right\rbrace.
	\label{Eq: sum1}
	\end{equation}
	
	In the following discussion, we calculate Eq. \eqref{Eq: sum1}.
	Let the $K$ roots $x_n$ of the equation $\eta = S_{K, 1}x_n$ be denoted 
	$\xi_i, i=0,\cdots,K-1$.
	Because the map $S_{K, 1}(x)$ corresponds to the $K$-angle formula of the cot function,
	$\eta$ is given by
	\begin{equation}
	\begin{array}{lll}
		\eta &=& x_{n+1} = S_{K,1}(x_n)\\
		 &=& K\frac{x_n^K+\mbox{($K-2$ th and the smaller order terms)}}{Kx_n^{K-1}+\mbox{($K-3$ th and the smaller order terms)}}
	\end{array}
	\end{equation}
	where $x_{n+1}= \eta$. Then it follows that
	\begin{equation}
	x_n^K-\eta x_n^{K-1}+(\mbox{$K-2$ th and the smaller order terms})=0.
	\end{equation}
	Because by definition $\xi_i$ is a root of the above $K$th-degree equation,
	it follows that $(x_n-\xi_0)(x_n-\xi_1)\cdots(x_n-\xi_{K-1})=0$.
	According to the relation between the roots and coefficients of a $K$th-degree equation,
	we have that
	\begin{equation}
	\eta = \sum_{m=0}^{K-1}\xi_m = \sum_{m=0}^{K-1}\cot\left\lbrace \frac{1}{K}\cot^{-1}\left(\frac{\eta}{K}\right)+\frac{m}{K}\pi\right\rbrace.
	\end{equation}
	Therefore, because
	\begin{equation}
	\left|S_{K, 1}^{-1}(0, \eta)\right|  = \eta,
	\end{equation}
	Eq. \eqref{Ap Eq lebesgue} holds.
	
	(II) In the case of $\alpha=-1$.
	
	Consider the case $\eta>0$ as in (I).
	Because the map $S_{K,-1}$ decreases monotonically,
	{\small
	\begin{equation}
	\begin{array}{lll}
	&&\left|S_{K, -1}^{-1}(0, \eta)\right|\\
	 &=& \displaystyle \sum_{m=0}^{K-1}\left[
	\cot\left(\frac{\pi}{2K}+\frac{m}{K}\pi\right)-
	\cot\left\lbrace \frac{1}{K}\cot^{-1}\left(\frac{-\eta}{K}\right)+\frac{m}{K}\pi\right\rbrace
	\right]\\
	&=&-\displaystyle \sum_{m=0}^{K-1}\cot\left\lbrace \frac{1}{K}\cot^{-1}\left(\frac{-\eta}{K}\right)+\frac{m}{K}\pi\right\rbrace.
	\end{array}
	\end{equation}
}
	For the map $S_{K,-1}$, the following relation holds:
	\begin{equation}
	\begin{array}{lll}
	\eta = -K\frac{x_n^K+\mbox{($K-2$ th and the smaller order terms)}}{Kx_n^{K-1}+\mbox{($K-3$ th and the smaller order terms)}}\\
	x_n^K + \eta x_n^{K-1}+(\mbox{$K-2$ th and the smaller order terms})=0.
	\end{array}
	\end{equation}
	According to the relation between the roots and coefficients of a $K$th-degree equation,
	we have the relation:
	\begin{equation}
	\begin{array}{rll}
	-\eta = \displaystyle\sum_{m=0}^{K-1}\xi_m = \sum_{m=0}^{K-1}\cot\left\lbrace \frac{1}{K}\cot^{-1}\left(\frac{-\eta}{K}\right)+\frac{m}{K}\pi\right\rbrace\\
	\therefore \displaystyle-\sum_{m=0}^{K-1}\cot\left\lbrace \frac{1}{K}\cot^{-1}\left(\frac{-\eta}{K}\right)+\frac{m}{K}\pi\right\rbrace = \eta.
	\end{array}
	\end{equation}
	Therefore, it follows that
	\begin{equation}
	\left|S_{K, -1}^{-1}(0, \eta)\right| = 
	\displaystyle-\sum_{m=0}^{K-1}\cot\left\lbrace \frac{1}{K}\cot^{-1}\left(\frac{-\eta}{K}\right)+\frac{m}{K}\pi\right\rbrace = \eta
	\end{equation}
	and Eq. \eqref{Ap Eq lebesgue} holds.
\end{proof}

At $\alpha= \pm 1$, the SGB transformations preserve the Lebesgue measure
for any $K \geq 2$. Thus, for the SGB transformations
the measure for the whole set cannot be normalized to unity. Then we define the ergodicity for the system with the infinite measure as follows
(according to Definition \ref{Def:ergodicity}):
\begin{defi}[ergodicity \cite{lasota2008probabilistic}] \label{Def:ergodicity}
	Let $(X, \mathscr{A}, \mu)$ be a measure space. $S:X \to X$ is a measurable transformation on the measure space $(X, \mathscr{A}, \mu)$.
	The transformation $S$ is called ergodic if every invariant set $A\in \mathscr{A}$ is such that either $\mu(A)=0$ or $\mu(X\backslash A)=0$.
\end{defi}

\begin{theo}\label{Theo:Infinit ergodicity}
	The SGB transformations at $\alpha=\pm 1$ are ergodic.
\end{theo}

\begin{proof}
	For the map $S_{K, \pm 1}$ substituting $\cot(\pi \theta_n)$ into $x_n \in \mathbb{R}\backslash B$, one has the induced map 
	$\displaystyle \bar{S}_{K, \pm 1}: X\overset{\rm def}{=}\frac{1}{\pi}\mbox{arccot}(\mathbb{R}\backslash B) \to X$ such that
	\begin{equation}
	\theta_{n+1} = \bar{S}_{K, \pm 1}(\theta_n) = \frac{1}{\pi}\cot^{-1}\left\lbrace \pm K \cot(\pi K\theta_n)\right\rbrace.
	\end{equation}
	The Figure \ref{Fig:transform} shows the relation between $\mathbb{R}\backslash B$ and $X$ in the range of $-10<x_n<10$.
	\begin{figure}[h]
	    \centering
	    \includegraphics[width=.8\columnwidth]{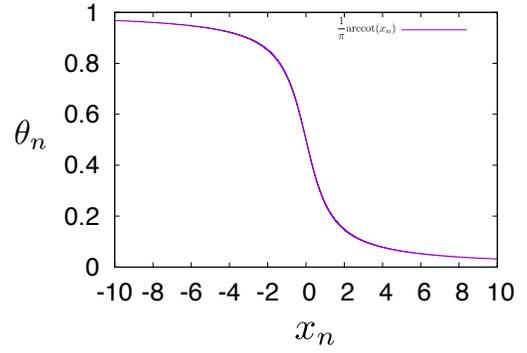}
	    \caption{The Relation between $x_n$ and $\theta_n$.}
	    \label{Fig:transform}
	\end{figure}
	
	Since we eliminate countably infinite number of points whose measure is 0 from $(0,1)$ to obtain the set $X$, $X \subset (0,1)$.
     The map $\bar{S}_{K, \pm 1}$ has topological 
	conjugacy with the map $S_{K, \pm 1}$, so that
	the ergodic properties of $\bar{S}_{K, \pm 1}$ are the same as those of $S_{K, \pm 1}$.
	In terms of absolute value of the derivative of $\bar{S}_{K, \pm 1}$, it holds that
	\begin{equation}
	\left|\bar{S}_{K, \pm 1}'(\theta)\right| = \frac{K^2\left\lbrace 1+\cot^2(\pi K\theta)\right\rbrace }{K^2\cot^2(\pi K \theta)+1}>1, 
	{}^\forall \theta \in X.
	\end{equation}
	Take the contraposition for Definition \ref{Def:ergodicity} and we will show that 
	\begin{equation}
	\begin{array}{lll}
	&&\mbox{for any}~ A~\mbox{s.t.}~\mu(A) \neq 0,~ \mbox{and}~ \mu(A^c) \neq 0\\
	 &\Rightarrow& ~\mbox{$A$ is not invariant}. 
	\end{array}\label{Eq: complement}
	\end{equation}
	In a way similar to the proof of the mixing property in generalized Boole transformations \cite{umeno2016exact} and
	the exactness in super generalized Boole transformations \cite{okubo2018universality}, 
	we define the open intervals $\left\lbrace I_{j,n}\right\rbrace $ for which the following relations hold:
	\begin{equation}
	\begin{array}{rll}
	I_{j, n} &\subset& \left(\eta_{j, n}, \eta_{j+1, n}\right),~ \eta_{j, n} < \eta_{j+1,n},\\
	n &\in& \mathbb{N},\\
	 0 &\leq& j \leq K^n-1,\\
	\eta_{0, n} &=& 0~~\mbox{and}~~\eta_{K^n,n}=1,\\
	\bar{S}_{K, \pm 1}^n (I_{j, n}) &=& X.
	\end{array}
	\end{equation} 
	Figure \ref{Fig: Infinite ergodic interval} illustrates the case of $\{I_{j,1}\}$ 
	for $K=3, 4,$ and $5$ at $\alpha=1$.
	\begin{figure}[H]	
	    \begin{minipage}{\columnwidth}
		    \centering
		    \includegraphics[width=.6\columnwidth]{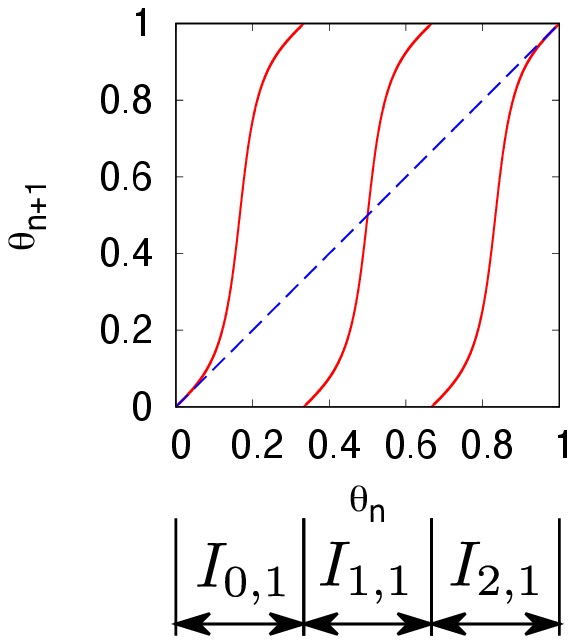}
		    \subcaption{$K=3$}
	    \end{minipage}
	    \begin{minipage}{\columnwidth}
		    \centering
		    \includegraphics[width=.6\columnwidth]{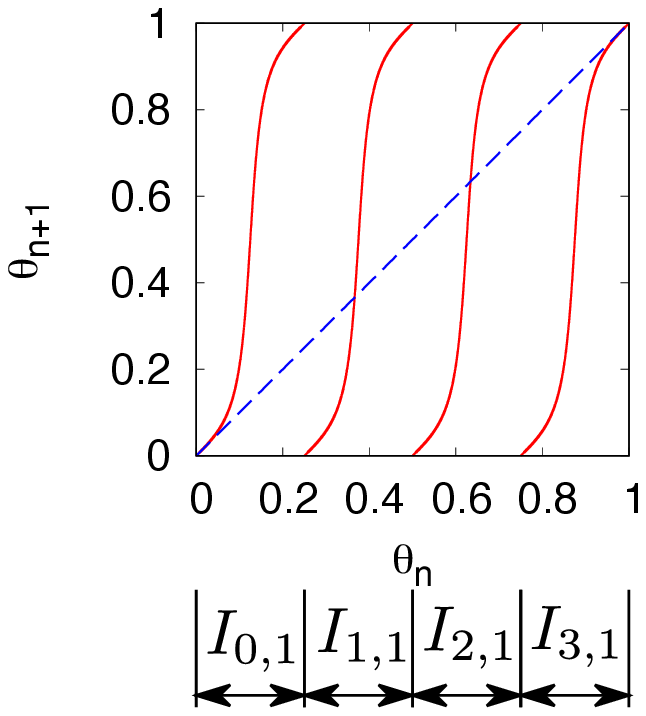}
		    \subcaption{$K=4$}
		\end{minipage}
		\begin{minipage}{\columnwidth}
		    \centering
	    	\includegraphics[width=.6\columnwidth]{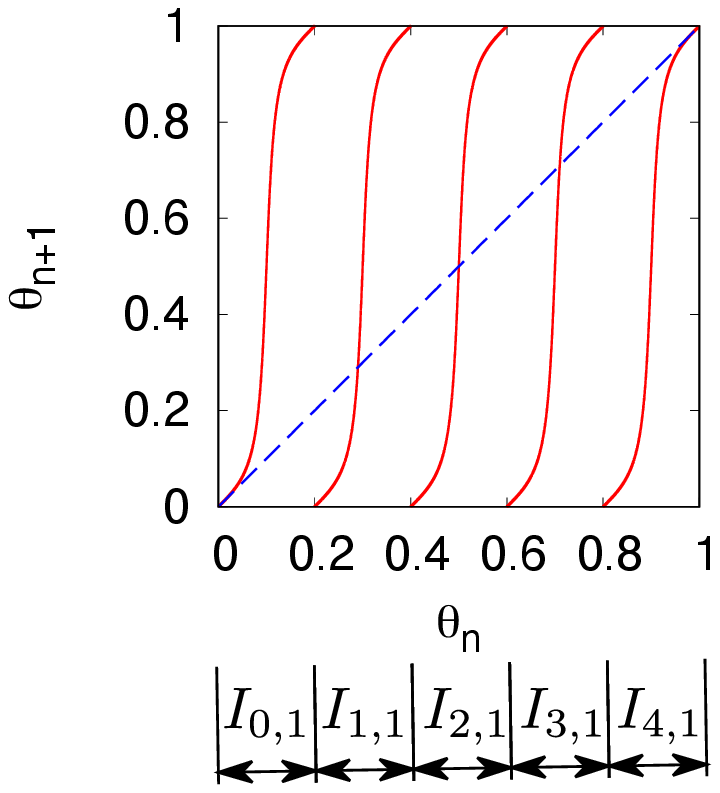}
	    	\subcaption{$K=5$}
	    \end{minipage}
		\caption{
				Solid lines correspond to the transformation $\bar{S}_{K, 1}$ which has exact topological conjugacy with the super generalized Boole transformation $S_{K, 1}$, where $K=3,4$ and $5$. Dashed line corresponds to the line $\theta_{n+1}= \theta_n$. 
			}
			\label{Fig: Infinite ergodic interval}
	\end{figure}
		
		Since the absolute value of the derivative $\bar{S}_{K, \pm 1}'$ on any $I_{j,n}$ is larger than 
		( $\because \{\theta|\cot^2(\pi K\theta)=\infty\} \notin X$) unity,
		the length of the interval $I_{j,n}$ becomes infinitesimal as $n \to \infty$.
		Then, for any set $A$ such that $\mu(A) \neq 0$, it follows that
		\begin{equation}
		{}^\exists p,q~\mbox{s.t.} I_{p,q} \subset A.
		\end{equation}
		From the definition of $I_{p,q}$, it follows that
		\begin{equation}
		\begin{array}{lll}
		\bar{S}_{K, \pm 1}^q I_{p,q} &=& X,\\
		\therefore \bar{S}_{K, \pm 1}^q A &=& X.
		\end{array}
		\end{equation}
		Next, for any set $A$ such that $\mu(A^c) \neq 0$, it follows that
		\begin{equation}
		A \neq X.
		\end{equation}
		Then, for any $A$ such that $\mu(A) \neq 0$ and $\mu(A^c) \neq 0$, it follows that
		\begin{equation}
		{}^\exists q \in \mathbb{N}~\mbox{s.t.}~ \bar{S}_{K, \pm 1}^q A = X~\mbox{and}~ A\neq X.
		\end{equation}
		This means that the set $A$ is \textit{not} invariant. 
		Therefore, Theorem \ref{Theo:Infinit ergodicity} holds.
	\end{proof}
	
\section{Normalized Lyapunov exponent}
According to the Darling-Kac-Aaronson theorem \cite{aaronson1997introduction}, for infinite measure $m$, for a conservative, ergodic, measure presrving map $T$ 
and for a function $f$ such as $f\in L^1(m), f\geq0, \int_Xfdm>0$ where $X$ is a set on which the map $T$ is defined,
normalized time average of $f$ converges to the normalized \textit{Mittag-Leffler distribution} such as \cite{akimoto2010subexponential,akimoto2010role,akimoto2015generalized}
\begin{equation}
\frac{1}{a_n}\sum_{k=0}^{n-1}f\circ T^k \to \left(\int_X fdm\right)Y_\gamma,
\end{equation}
where $a_n$ is the return sequence and $Y_\gamma$ is a random variable which obeys the normalized Mittag-leffler distribution of order $\gamma$.
In the case of the Boole transformation, the return sequence is obtained as $a_n = \frac{\sqrt{2n}}{\pi}$ \cite{aaronson1997introduction}.

In the case of this SGB transformations at $\alpha= \pm 1$, consider $f$ as $\log\left|\frac{dS_{K, \pm1}}{dx}\right|$
and we clarify whether the notmalized Lyapunov exponent converges to the normalized Mittag-Leffler distribution
by numerical simulation.

We have that $\log\left|\frac{dS_{K, \pm1}}{dx}\right| \geq 0$\cite{okubo2018universality}.
In the folloing, we assume such condition as
\begin{equation}
\begin{array}{lll}
a_n &\propto& n^\frac{1}{2},\\
\log\left|S_{K, \pm 1}'\right| &\in& L^1(\mu).
\end{array}
\end{equation}
as the case $(K, \alpha) = (2, 1)$ \cite{aaronson1997introduction}.

We calculate the normalized Lyapunov exponents such as
\begin{equation}
\lambda = \frac{c(K)}{\sqrt{n}}\sum_{i=0}^{n-1}\log\left|\frac{dS_{K, \pm 1}}{dx}(x_i)\right|
\end{equation}
where $c$ are the normalization constants to make the mean values equal to unity.
Figures \ref{Fig:normalized Lyapunov K=3 alpha=1}, \ref{Fig:normalized Lyapunov K=4 alpha=1},
\ref{Fig:normalized Lyapunov K=5 alpha=1} \ref{Fig:normalized Lyapunov K=3 alpha=-1},
\ref{Fig:normalized Lyapunov K=4 alpha=-1} and \ref{Fig:normalized Lyapunov K=5 alpha=-1} show the density function of the normalized Lyapunov exponents 
for $(K, \alpha) = (3,1), (4,1)$, $(5,1)$, $(3,-1)$, $(4,-1)$ and $(5,-1)$, respectively,
which confirms that their normalized Lyapunov exponents are \textit{distributed} according to the normalized Mittag-Leffler distribution of order $\frac{1}{2}$.

\begin{figure}[H]
	\begin{minipage}{\columnwidth}
    	\centering
	    \includegraphics[width=.75\columnwidth]{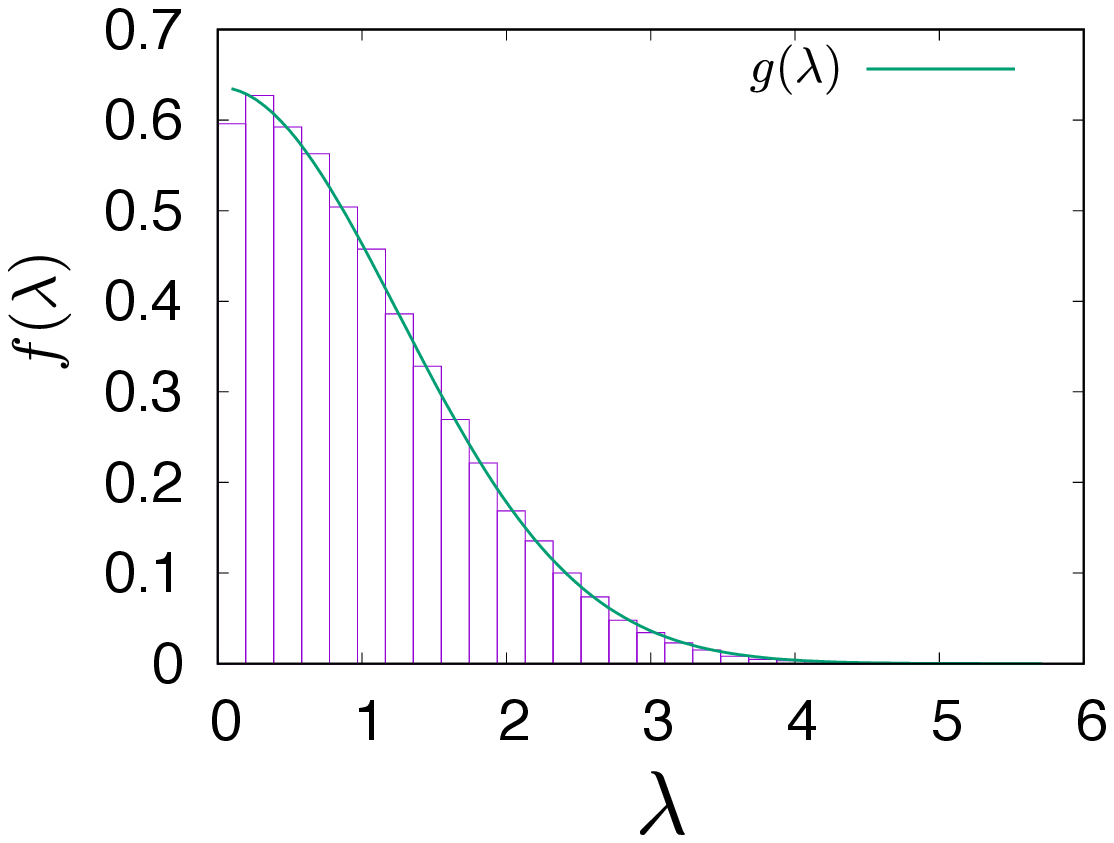}
	    \subcaption{$K=3$}
    	\label{Fig:normalized Lyapunov K=3 alpha=1} 
	\end{minipage}\\
	\begin{minipage}{\columnwidth}
	    \centering
	    \includegraphics[width=.75\columnwidth]{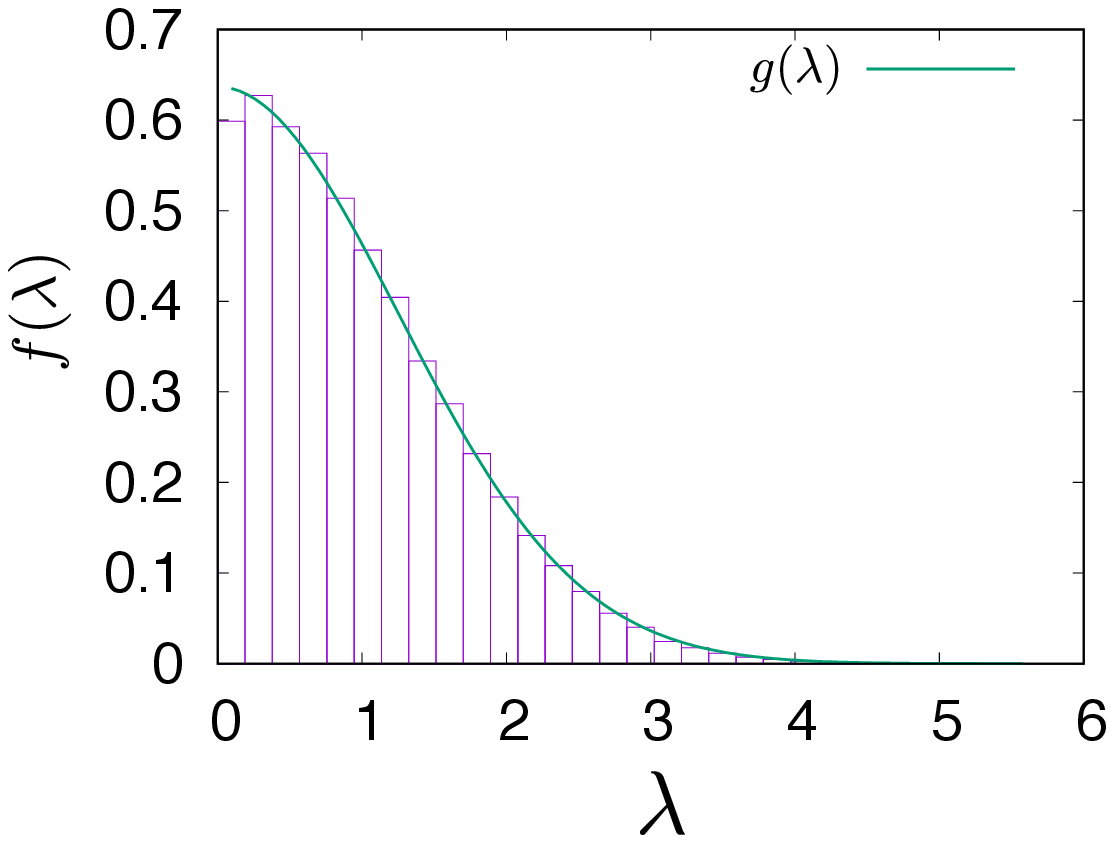}
		\subcaption{$K=4$}
		\label{Fig:normalized Lyapunov K=4 alpha=1} 
	\end{minipage}\\
	\begin{minipage}{\columnwidth}
		\centering
		\includegraphics[width=.75\columnwidth]{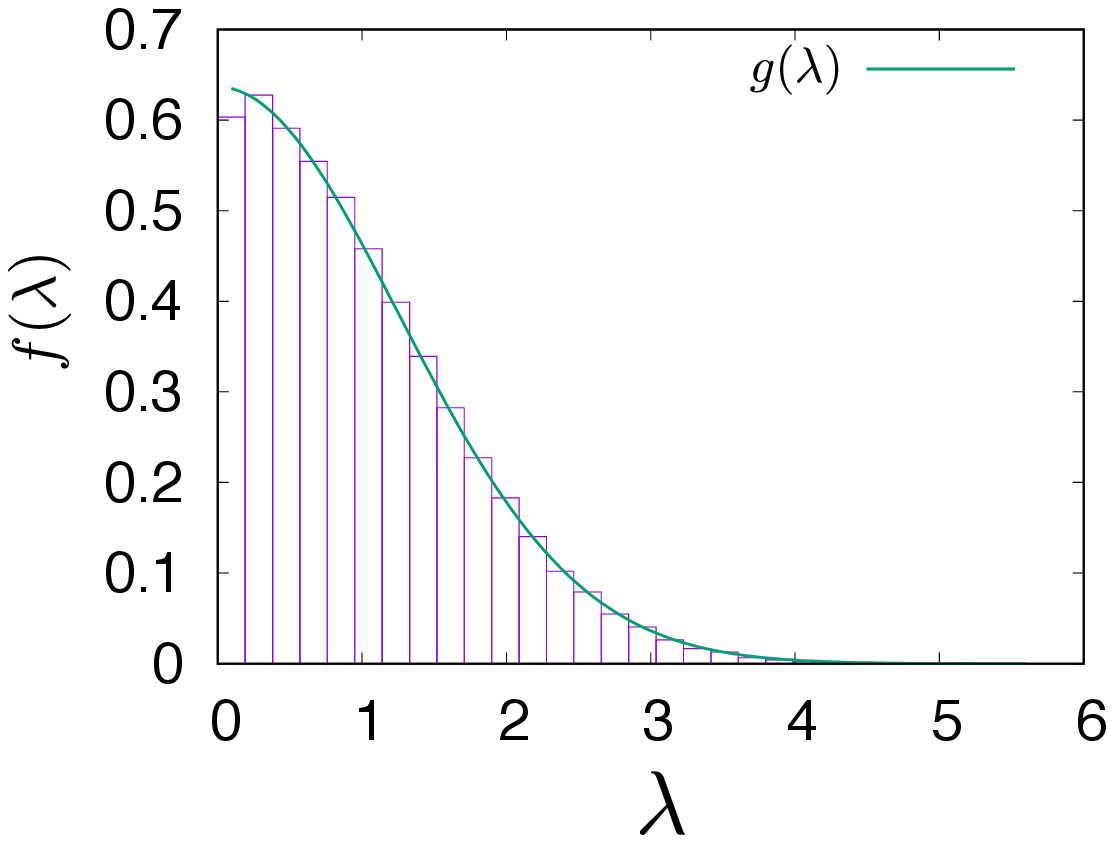}
		\subcaption{$K=5$}
		\label{Fig:normalized Lyapunov K=5 alpha=1} 
	\end{minipage}
	\caption{Relation between the density functions of normalized Lyapunov exponent and normalized Lyapunov exponent in SGB transformation for $K=3, 4 ~\mbox{and}~ 5 (\alpha=1)$.
	The number of initial points is $M=10^5$ and the number of iteration is $N=10^5$.
	Initial points are distributed to obey the normal distribution whose mean and variance are 0 and 1, respectively. The bar graph represents the numerical simulation of the normalized Lyapunov exponents and the solid line represents the normalized Mittag-Leffler distributions of order $\frac{1}{2}$.}
\end{figure}

\begin{figure}[H]
	\begin{minipage}{\columnwidth}
    	\centering
	    \includegraphics[width=.75\columnwidth]{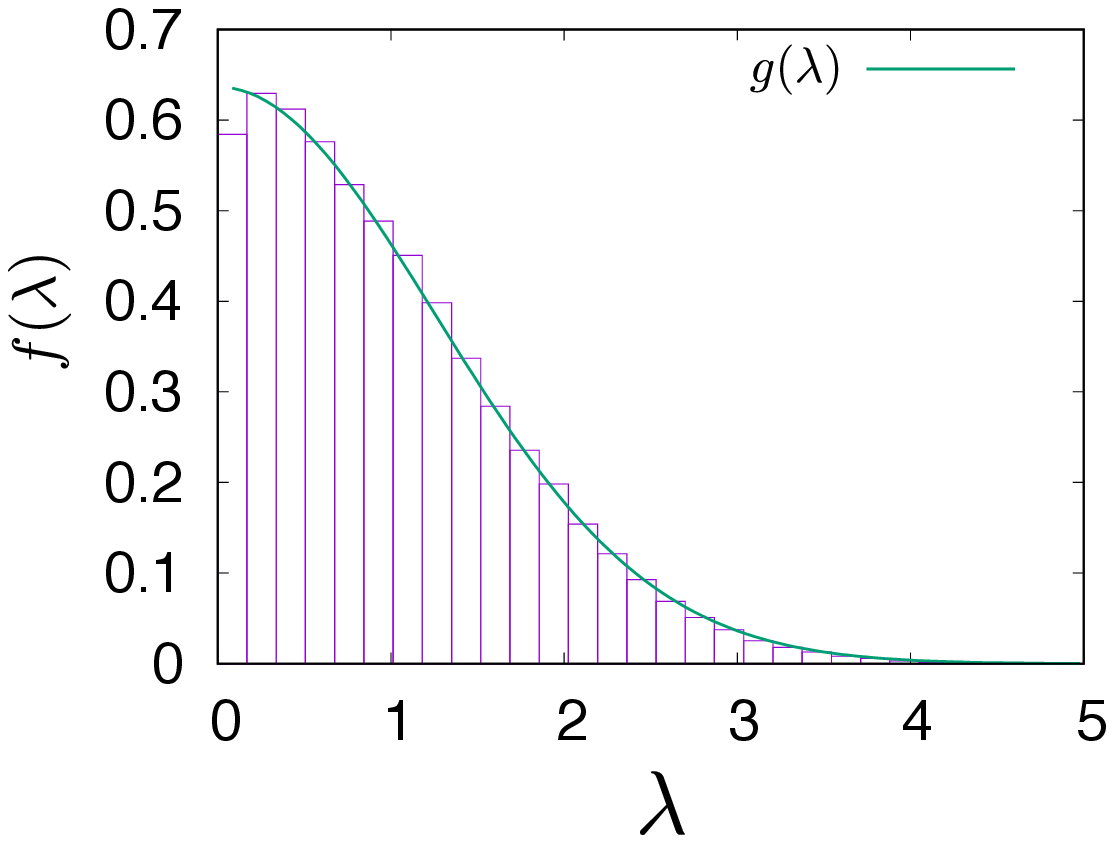}
	    \subcaption{$K=3$}
    	\label{Fig:normalized Lyapunov K=3 alpha=-1} 
	\end{minipage}\\
	\begin{minipage}{\columnwidth}
	    \centering
	    \includegraphics[width=.75\columnwidth]{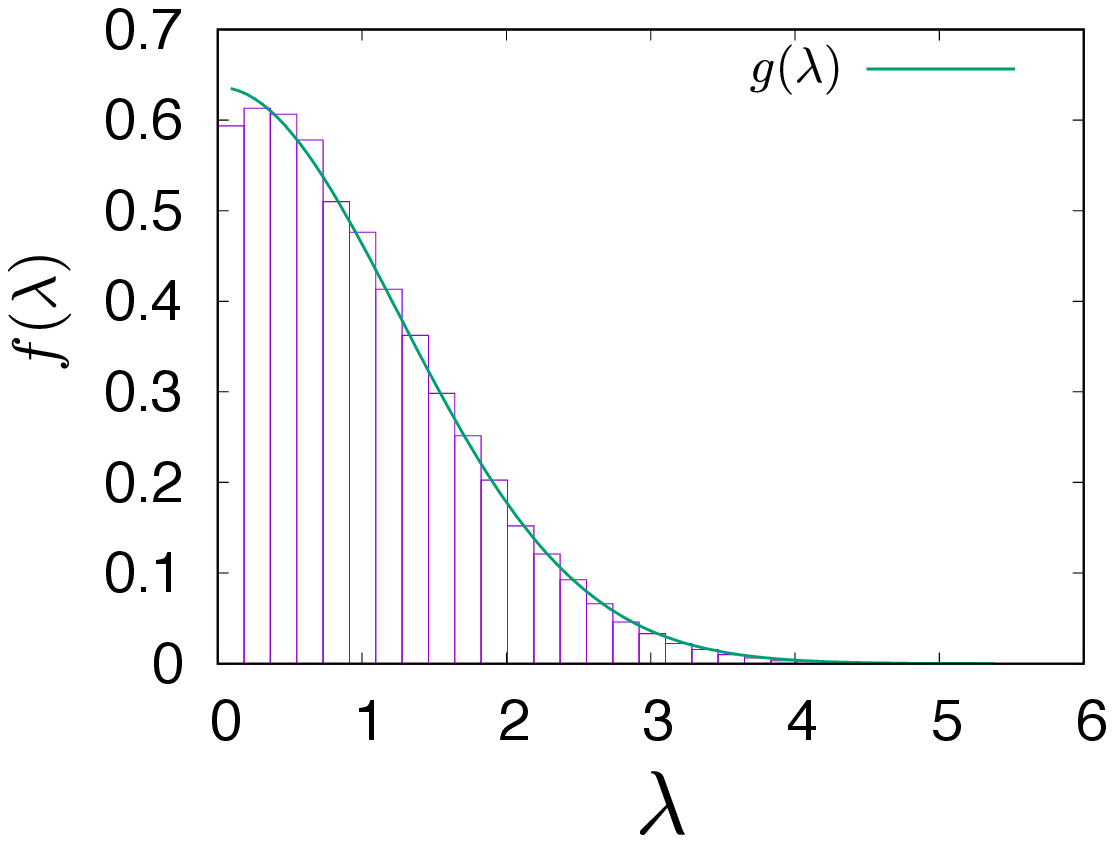}
		\subcaption{$K=4$}
		\label{Fig:normalized Lyapunov K=4 alpha=-1} 
	\end{minipage}\\
	\begin{minipage}{\columnwidth}
		\centering
		\includegraphics[width=.75\columnwidth]{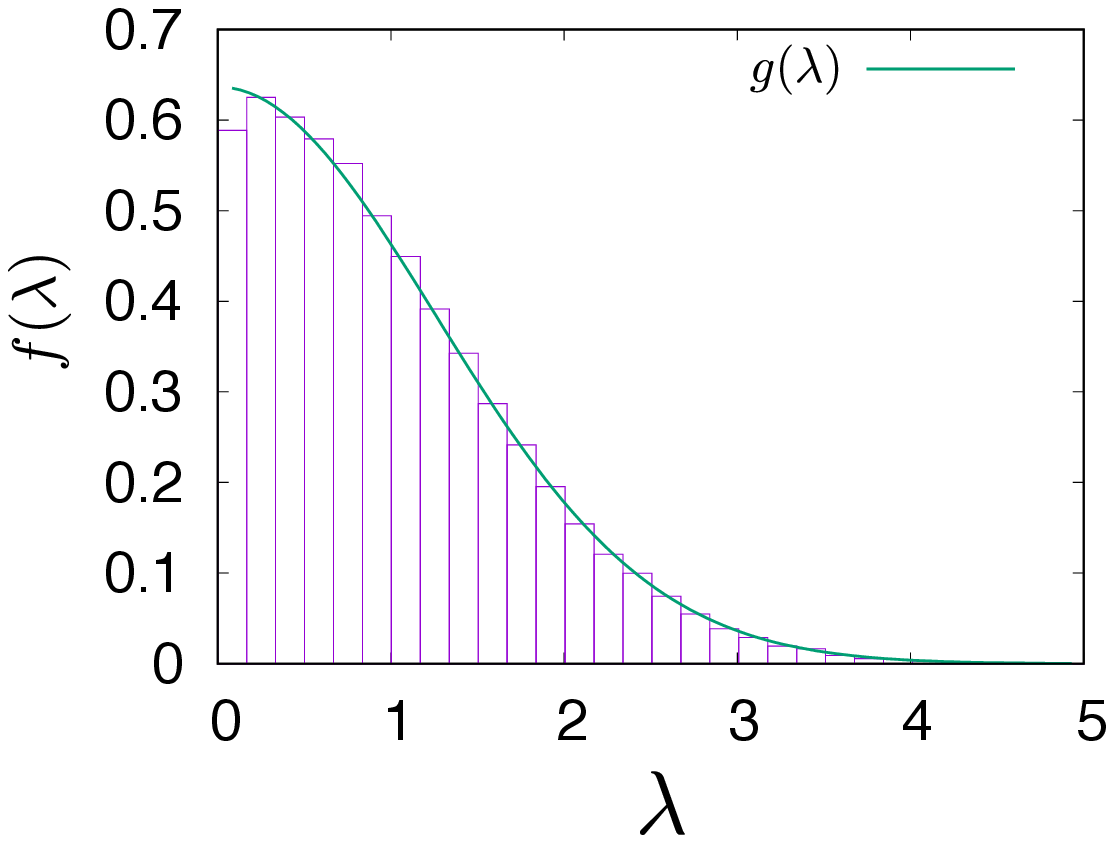}
		\subcaption{$K=5$}
		\label{Fig:normalized Lyapunov K=5 alpha=-1} 
	\end{minipage}
	\caption{Relation between the density functions of normalized Lyapunov exponent and normalized Lyapunov exponent in SGB transformation for $K=3, 4 ~\mbox{and}~ 5 (\alpha=-1)$.
	The number of initial points is $M=10^5$ and the number of iteration is $N=10^5$.
	Initial points are distributed to obey the normal distribution whose mean and variance are 0 and 1, respectively. The bar graph represents the numerical simulation of the normalized Lyapunov exponents and the solid line represents the normalized Mittag-Leffler distributions of order $\frac{1}{2}$.}
\end{figure}

Figure \ref{Fig: c and K} shows the relation between normalization constant $c(K)$ and $K$ at $\alpha=\pm 1$. 
We can see that $c(K)$ tends to decrease as $K$ increases.
At $(K,\alpha)=(2,1)$, we know that $c(K) = \frac{1}{2\sqrt{2}} \simeq 0.354$ from $a_n = \frac{\sqrt{2n}}{\pi}$ \cite{akimoto2010subexponential}. Figure \ref{Fig: c and K} is consistent with this result and from the fact that
the points at $(K,\alpha)=(2,-1), (3,1)$ and $(3,-1)$ are on $g(K)$ and that $\int\ln|S'_{2,-1}(x)|dx=\int\ln|S'_{3,\pm1}(x)|dx = 2\pi$, we conjecture 
that for $S_{2,-1}$, the return sequence $a_n$ is given by $a_n = \frac{\sqrt{2n}}{\pi}$
and that  
for $S_{3,\pm1}$, $a_n = \frac{\sqrt{3n}}{\pi}$.

\begin{figure}[H]
    \centering
    \includegraphics[width=\columnwidth]{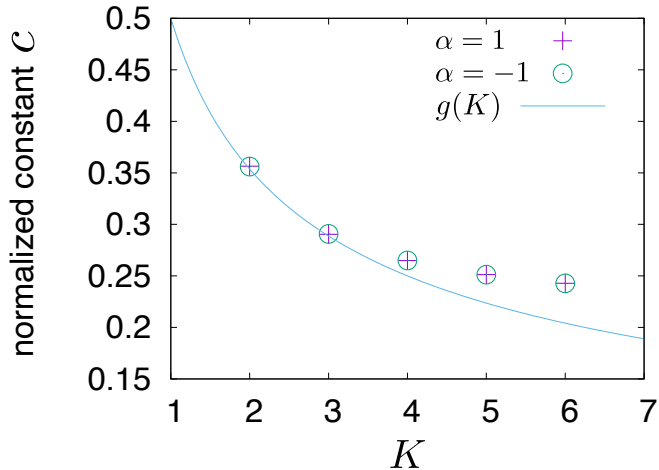}
    \caption{The relation between normalization constant $c(K)$ and parameter $K$. The function $g(K)$ is rewritten as $g(K) = \frac{1}{2\sqrt{K}}$.} 
    \label{Fig: c and K}
\end{figure}

\section{Conclusion}
In this paper, we showed the statistical ergodic property of one dimensional chaotic maps, the super generalized Boole (SGB) transformations $S_{K, \alpha}$ at $\alpha=\pm1$. That is, for infinite number of $K$, we proved that the $S_{K, \pm 1}$ preserve the Lebesgue measure and that the dynamical systems are \textit{ergodic} for $K \geq 2$.
In the case of $K=2$ (the Boole transformation), Adler and Wiss proved its ergodicity in unbounded region \cite{adler1973ergodic} but
in our method, we proved the ergodicity by transforming the unbounded domain to the bounded domain using topological conjugacy.
In the previous work \cite{okubo2018universality}, the authors proved that the SGB transformations are \textit{exact} for $0<\alpha<1$($K=2N$) or $\frac{1}{K^2}<\alpha<1$ ($K=2N+1$), $N\in \mathbb{N}$ and they are dissipative for $\alpha>1$. The result of this paper connects these two regions in the same way of the generalized Boole transformations \cite{umeno2016exact}.
Then, we demonstrated that the normalized Lyapunov exponents actually obey the Mittag-Leffler distribution of order $\frac{1}{2}$
for $(K, \alpha) = (3, 1), (4,1)$, $(5, 1)$, $(3,-1)$, $(4,-1)$ and $(5,-1)$.
In this numerical experiments, the form of Mittag-Leffler distribution does not depend on the value of $K$ although
there is a relation between $c$ and $K$.
 Owing to these results, we obtain a class of countably infinite number of critical maps in the sense of \textit{Type} 1 or \textit{Type} 3 intermittency
which preserve the Lebesgue measure and are proven to be ergodic with respect to the Lebesgue measure.

In previous works various indicators were proposed to characterize the instablity when the corresponding Lyapunov exponent is zero
such as generalized Lyapunov exponent \cite{korabel2009pesin,korabel2010separation}and Lyapunov pair \cite{akimoto2010subexponential}.
It is fully expected that  the these infinite critical SGB transformations will be used as represented indicator maps in order to detect chaotic criticality since the ergodic properties are exactly obtained. 

\begin{acknowledgments}
One of the author, Ken-ichi Okubo, thanks to Dr. Takuma Akimoto for his fruitful advice.
\end{acknowledgments}

\bibliography{infinite-SGB2}

\end{document}